\documentclass{article}
\usepackage{amssymb,amsmath,ifmath,ifart,ifmisc,palatino,graphics,hyperref}
\usepackage[slantedGreek]{mathpazo}
\DeclareMathOperator{\Ibes}{I} 
\DeclareMathOperator{\Abes}{A} 
\DeclareMathOperator{\Ln}{Ln} 
\DeclareMathOperator{\rng}{rng} 
\newcommand{\vz}[1]{{\vv z}_{#1}} 
\newcommand{\vzd}[1]{{\dot{\vv z}}_{#1}} 
\newcommand{\vzb}[1]{\bar{{\vv z}}_{#1}} 
\newcommand{\dvz}{\md\upsilon(\vv z)} 
\newcommand{\odn}{\hat{\od}{[\vv\dr]}} 
\newcommand{\odnh}{\hat{\od}{[\hat{\vv\dr}]}} 
\newcommand{\odh}{\hat\od} 
\newcommand{\dreq}{\dr_{\mathrm{eq}}} 
\newcommand{\cphi}{\Phi} 
\newcommand{\cpsi}{\Psi} 
\newcommand{\eps}{\epsilon} 
\newcommand{\odc}{c} 
\DeclareMathOperator{\sO}{{\scriptstyle\mathcal O}} 
\newcommand{\separate}{\noindent\raisebox{0.5ex}{\rule{4in}{0.5pt}}} 
\newcommand{\mybox}[3]{\boxed{\raisebox{#1}{\rule{0pt}{#2}}\quad #3\quad}} 
\begin{document}
%
%
%
\iftitle{Diffusive transport in two-dimensional nematics}
\ifauthor{Ibrahim Fatkullin}{Department of Mathematics, University of Arizona, Tucson, AZ 85721, USA}
\ifauthor{Valeriy Slastikov}{Department of Mathematics, University of Bristol, Bristol BS8 1TW, UK}

\ifdate

\ifabstract{We discuss a dynamical theory for nematic liquid crystals describing the stage of evolution in which the hydrodynamic fluid motion has already equilibrated and the subsequent evolution proceeds via diffusive motion of the orientational degrees of freedom. This diffusion induces a slow motion of singularities of the order parameter field. Using asymptotic methods for gradient flows, we establish a relation between the Doi-Smoluchowski kinetic equation and vortex dynamics in two-dimensional systems. We also discuss moment closures for the kinetic equation and Landau-de Gennes-type free energy dissipation.}
%
%
%
%
\section{Introduction} 

Dynamics of liquid crystalline systems is traditionally described in a framework of theories combining fluid dynamics equations, constitutive relations between the hydrodynamic stress tensor and liquid crystalline order parameters, and  evolution equations for the latter \cite{BerEdw}, \cite{DoiEdw}, \cite{Erick61}, \cite{Lesl68}. In the absence of hydrodynamic motion, the relaxation of the orientational degrees of freedom is induced by the free energy dissipation. This relaxation is generally slow and can be characterized by the evolution of topological defects in the order parameter fields. Our goal in this work is to derive equations of motion for these defects starting from a kinetic Doi-Smoluchowski-type equation \cite{DoiEdw}. To accomplish this task we use asymptotic methods for gradient flows similarly to the way it is used in the Ginzburg-Landau theory \cite{Beth94}, \cite{E94}, \cite{Lin96}, \cite{Neu90}, \cite{San04}. We omit most of the technical details concentrating rather on the methodology and final results. Additionally, we discuss the possibility of describing the dissipative dynamics in terms of the order parameter fields, similar, to, e.g., Landau-de Gennes theory \cite{dGen95}. We also discuss the extent to which the formal moment closures provide the correct evolution equations.

%

We choose the Doi-Smoluchowski (D-S) model \cite{DoiEdw} as the starting point for our analysis, because, in some respect, it is a microscopic theory, in comparison to, e.g., Ericksen-Leslie \cite{Erick61}, \cite{Lesl68}, or Beris-Edwards models  \cite{BerEdw}.  In the D-S theory, the state of a liquid crystalline system is described by means of a probability density of rods orientations; and the D-S equations are kinetic equations for this density. The other aforementioned models are based on description via its various moments, and should, in principle, be derived from a D-S-type model. Mathematically, the D-S equations describe gradient flow dynamics for the Onsager-Maier-Saupe free energy \cite{Fat07a} in Wasserstein metric \cite{Vill03}. This makes analysis of the system amenable to methods of the theory of gradient flows \cite{Ambro05}.

In this paper we are interested in the two-dimensional model.  One of the characteristic features of two-dimensional systems is that the energy of topological singularities, or {\em vortices,} diverges logarithmically in meaningful asymptotic limits. The Doi dynamics of orientation density reduces directly to the vortex dynamics. Due to this, it is impossible to find a nontrivial regime in which a Landau-de Gennes-type gradient flow evolution for the second moment can be derived from the D-S model. This is different from the three-dimensional theory, where one can reduce the D-S dynamics to equations for the second moment \cite{E06}, \cite{Wang13}.

\paragraph{The paper is organized as follows:} We start by reviewing the two-dimensional spatially extended Onsager-Maier-Saupe free energy, introduced in our earlier works \cite{Fat07a}, \cite{Fat09a}. Understanding the landscape of this free energy provides us with characterization of the states which posses ``moderate'' amounts of free energy, as compared to the ground, uniform nematic state; we call such states {\em tempered.} These states are uniquely characterized by the location of vortices, and some auxiliary function, so that the evolution of tempered states may be completely described by the evolution of these quantities. We then set up the D-S dynamics as a gradient flow dynamics for our free energy, and carry out an asymptotic reduction. In order to explain the methods and ideas of this reduction we consider a finite-dimensional example on a rather rigorous level. After that, we implement an analogous procedure for the infinite-dimensional system, deriving equations governing the vortex motion. Finally, we derive an infinite hierarchy of equations for moment of the orientation density, and discuss possible closures.

\section{Review of the spatially-extended Onsager-Maier-Saupe model}
The main goal of the next two sections is to familiarize the reader with spatially extended Onsager-Maier-Saupe model \cite{Fat07a,Fat09a}, and to state the principal results of this paper in its context. As mentioned in the Introduction, we specialize to two-dimensional systems.   

In the framework of this model, the state of a liquid crystalline system is characterized by the space-dependent orientation probability density of nematic molecules, $\od(\varphi,\vv z)$ integrating to unity over $\varphi$ for each $\vv z\in\dom$. Here $\varphi\in[0,2\pi)$ is the orientation parameter of liquid crystalline molecules, and $\vv z\in\dom\subset\C$ is a spatial variable. Note that we employ complex notation $\vv z=x+\mi y$ for the spatial coordinates, as this simplifies many calculations. Refer to appendix for additional details. 

\paragraph{The free energy}
of the liquid crystalline system, $\eNrg(\od)$, is a functional of orientation probability density $\od(\varphi,\vv z)$ and is represented as an integral over the spatial domain $\dom$ of the sum of two contributions:
\begin{equation}\label{eq:fnrg2d}
	\eNrg(\od)\,=\,\int_\dom\Big[\,\fNrg_\ore(\od)+\fNrg_\nloc(\od)\,\Big]\dvz.
\end{equation}
Hereafter we use $\dvz=\md x\md y$ to denote the volume element in $\dom$.

The {\em orientational free energy density\,} $\fNrg_\ore$ is an Onsager-type functional,
\begin{align}\label{eq:ons2d}
	\fNrg_\ore(\od)=\int_0^{2\pi}\od(\varphi,\vv z)\,\ln\!\big[2\pi\od(\varphi,\vv z)\big]\,\md{\varphi}\,-
	\frac{\gamma}{2}\iint_0^{2\pi}\cos2(\varphi-\varphi^\prime)\,\od(\varphi,\vv z)\od(\varphi^\prime,\vv z)\,
	\md{\varphi}\md{\varphi^\prime}+\,C_{\gamma},
\end{align}
where the constant $C_{\gamma}$ is chosen to have $\fNrg_\ore\geq0$. The positive parameter $\gamma$ is referred to as {\em concentration}.

The {\em elastic free energy density} is a quadratic functional of the {\em order-parameter} field equivalent to that of the Landau-de~Gennes theory:
\begin{align}\label{eq:els_iso}
	\fNrg_\nloc(\od)\,=\,\frac{\eps^2}{2}\,|\!\grd\vv{\dr}(\vv z)|^2.
\end{align}
Here  $\vv{\dr}(\vv{z})$ is the order parameter field  related to the orientation probability density function, $\od$, via
\begin{equation}\label{eq:dir_od}
	\vv{\dr}(\vv{z})\,=\,\int_0^{2\pi}\,\me^{\,2\mi\varphi}\od(\varphi,\vv z)\md\varphi.
\end{equation}
The positive parameter $\eps^2$ in equation (\ref{eq:els_iso}) is called the {\em elastic modulus}.

\paragraph{A useful observation} is that the total free energy (\ref{eq:fnrg2d}) may be decomposed in the following way:
\begin{align}\label{eq:nrg_decomp}
	\eNrg(\od)\,
	&=\,\int_\dom\Big[\,\frac{\eps^2}{2}\,\big|\!\grd\vv{\dr}\big|^2\,+\,W^{\gamma}(\dr)\Big]\md\vv z\,+\,\int_\dom\sNrg(\od|\odn)\dvz.
\end{align}
Here for a given probability density $\od$, the order parameter field, $\vv\dr$, is defined in \eqref{eq:dir_od}. The potential $W^{\gamma}$ is given by
\begin{equation}\label{eq:pot_n}
	W^{\gamma}(\dr)\,=\,-\,\frac{\gamma\dr^2}{2}\,+\,
	\big[\dr\Abes(\dr)\,-\,\ln\Ibes_0(\Abes(\dr))\big]\,+\,C_{\gamma},
\end{equation}
where we use notation $\dr =| \vv{\dr} |$. Some of the properties of $W^{\gamma}(\dr)$ and related special functions are presented in Appendix~\ref{sec:wgamma}. In particular, we must pick $\gamma>2$ to assure the existence of nematic states.

The {\em locally-equilibrated} probability density  $\odn$ in equation (\ref{eq:nrg_decomp}) is related to $\vv\dr$ via
\begin{equation}\label{eq:fam2d}
	\odn(\varphi,\vv z)=\frac{\exp\big\{\abes\!\big(\dr(\vv z)\big)\,\cos\big(2\varphi-\arg\vv\dr(\vv z)\big)\big\}}
	{2\pi\ibes_0\!\big(\!\abes(\dr)\big)};
\end{equation}
and $\sNrg(\od|\odn)$ is the relative entropy of $\od$ with respect to $\odn$, \ie,
\begin{equation}
	\sNrg(\od|\odn)\,=\,\int_0^{2\pi}\ln\frac{\od(\varphi)}{\odn(\varphi)}\;\od(\varphi)\md\varphi.
\end{equation}
The field $\odn$ has a straightforward interpretation: whenever $\od=\odn$,  $\sNrg(\od|\odn)=0$, and thus $\odn$ minimizes the total free energy in the class of all fields with prescribed order-parameter $\vv\dr(\vv z)$.

We define the {\em reduced free energy}\, $\nNrg(\vv{\dr})$  as a functional of the order parameter $\vv\dr$ alone,
\begin{equation}\label{eq:nrg_decompose}
	\nNrg(\vv{\dr})\,=\,\int_\dom\Big[\,\frac{\eps^2}{2}\,\big|\!\grd\vv{\dr}\big|^2\,+\,W^{\gamma}(\dr)\Big]\dvz
\end{equation}
and notice its similarity to the Ginzburg-Landau energy. Thus, from \eqref{eq:nrg_decomp} we see that the total liquid crystalline energy, $\eNrg(\od)$, is decomposed in the sum of a Ginzburg-Landau-type energy, $\nNrg(\vv{\dr})$, and the relative entropy  $\sNrg(\od|\odn)$. This allows us to obtain an asymptotic limit in which the orientation density, $\od$, becomes enslaved to its second moment, $\vv\dr$ via equation (\ref{eq:fam2d}).

\paragraph{Multi-vortex patterns} are configurations of the order parameter field, $\vv\dr$, which appear in the limits as $\gamma\to\infty$ or $\eps\to0$. In this work we are interested in the dynamics of liquid crystalline systems in the limit as $\eps \to 0$. Even though this particular limit was not considered in \cite{Fat09a}, the results below can be easily obtained by combining analysis in \cite{Fat09a} and results presented in \cite{San98} for the Ginzburg-Landau energy. 

Consider a family of order parameter fields, $\vv\dr^\epsilon(\vv z)$, which satisfy the boundary condition,
\begin{equation}\label{eq:bc_dir}
	\vv\dr^\epsilon(\vv z)\,=\,\dr_\eq\exp\{\mi\cpsi(\vv z)\},\quad\vv z\in\pd\dom,
\end{equation}
where $\dr_{\eq}$ is the minimizer of $W^\gamma (\dr)$. Assume that in the limit, as  $\eps \to 0$, the energy of these order parameter fields satisfies the bound
\begin{equation}\label{eq:tempered}
	\frac{1}{\eps^2}\nNrg(\vv\dr^\eps) \,\leq\, \pi \dr_{\eq}^2\,|d|\ln{\eps}+C,
\end{equation}
where $d=\deg\vv\dr^\eps|_{\pd\dom}$; $C$ is a positive constant independent of $\eps$. (Such configurations are called {\em tempered}.)
Then, as $\eps\to0$, $\vv\dr^\eps(\vv z)$, converge in appropriate sense (up to a subsequence) to
\begin{equation}\label{eq:th_lim_dir}
	\quad \vv\dr_*(\vv z)\,=\,\dreq\exp\big\{\mi\cphi(\vv z)\,+\,\mi\,\sgn d\sum_{k=1}^d\arg(\vv z-\vz{k})\big\},
\end{equation}
where $\cphi(\vv z)$ is a function with finite Dirichlet energy. Such a field $\vv\dr_*(\vv z)$ is a particular example of the so-called {\em multi-vortex field,} which in general may be represented as  
\begin{equation}\label{eq:th_lim_dir_gen}
	\vv\dr_*[\vz{1},\ldots,\vz{N};d_1,\ldots,d_N](\vv z)\,=\,\dreq\exp\bigg\{\mi\cphi(\vv z)\,+\,\mi\,\sum_{k=1}^Nd_k\arg(\vv z-\vz{k})\bigg\}.
\end{equation}
In this paper we allow the vortices to have different degrees $d_k=\pm1$. It is possible to show, that the results valid for tempered states $\vv\dr^\epsilon(\vv z)$ hold in the setting when vortices have degrees $d_k=\pm1$, provided $\vv\dr^\epsilon(\vv z)$ stays close to a multi-vortex configuration \eqref{eq:th_lim_dir_gen}. In this case we still refer to such states as tempered.  In particular, the limiting equations for the vortex dynamics remain valid until vortices of different signs approach each other (or the boundary) and undergo collision-annihilation process that we do not discuss here. Until that, the structure of all tempered states is completely characterized by the degrees and locations of vortices, $d_k$, $\vz{k}$, $k=1,\ldots,N$, and the function $\cphi(\vv z)$.


Our derivation of the limiting equations relies on the lower bound on the energy $\eNrg(\od^\eps)$. This bound follows from the results obtained in \cite{Fat09a} for Onsager-Maier-Saupe energy and in \cite{Alic12} for the Ginzburg-Landau energy:
\begin{equation}\label{eq:red_nrg_decomp}
	\frac{\eNrg(\od^\eps)}{\eps^2\dr_\eq^2}\,\,\geq\,\pi N\ln\eps\,+\,N\nNrg_0\,+\,
	\tilde\nNrg(\vz{1},\ldots,\vz{N};\cphi)\,+\,\frac{1}{\eps^2\dr_\eq^2}\,\int_\dom\sNrg\big(\od^\eps\big\vert\hat\od[\vv\dr_*]\big)\dvz.
\end{equation}
Here $E_0$ is a fixed constant related to the optimal profile problem, the {\em renormalized multi-vortex energy} is given by
\begin{equation} \label{eq:tildeE}
	\tilde\nNrg(\vz{1},\ldots,\vz{N};\cphi)\,=\,\frac{1}{2}\int_\dom|\grd\cphi(\vv z)|^2\dvz\,+\,U(\vz{1},\ldots,\vz{N}),
\end{equation}
where the first term is the Dirichlet energy of the field $\cphi(\vv z)$, which satisfies the boundary condition,
\begin{equation}\label{eq:psiphi}
	\cphi(\vv z)\,=\,\cpsi(\vv z)-\,\mi\,\sum_{k=1}^Nd_k\arg(\vv z-\vz{k}),\quad\vv z\in\pd\dom;
\end{equation}
and the second term is the {\em multi-vortex potential,}
\begin{align}\label{eq:th_nrg_vor_intro}
	U(\vz{1},\ldots,\vz{N})\,=\,&-\,\pi\sum_{j,\,k=1\atop j\neq k}^{N}d_kd_j\,\ln|\vz{k}-\vz{j}|\\\nonumber
	&+\,\sum_{k=1}^{N}d_k\oint_{\partial\dom} \ln|\vv z-\vz{k}|\md\cpsi (\vv z)  \,-\,
	\frac{1}{2}\sum_{j,\,k=1}^{N}d_kd_j\oint_{\partial\dom}\ln|\vv z-\vz{k}|\md\,\arg(\vv z-\vz{j}).
\end{align}
In this work we impose the Dirichlet boundary condition on the order parameter fields, $\vv\dr(\vv z)$, on $\pd\dom$, i.e., we prescribe the function $\cpsi(\vv z)$, cf. (\ref{eq:bc_dir}). Therefore the multi-vortex potential $U$ may be expressed explicitly as a function of vortex locations. This contrasts with a situation when the Neumann boundary condition is used. In that case, the multi-vortex potential also depends on the function $\cphi(\vv z)$. 
This may be seen if we rewrite the second term on the right-hand side of (\ref{eq:th_nrg_vor_intro}) in terms of $\cphi(\vv z)$ using the identity (\ref{eq:psiphi}) (as the function $\cpsi$ is not given). Therefore the multi-vortex potential, $U$, also depends on the boundary value of $\cphi(\vv z)$.

The energy decomposition (\ref{eq:red_nrg_decomp}) allow us to undertake an asymptotic reduction in the limit of small $\eps$. We see, that, as $\eps\to0$, the relative entropy term forces $\od(\varphi,\vv z)$ to remain close to $\odh[\vv\dr_*](\varphi,\vv z)$ at all times. Consequently, any gradient flow dynamics preserving the temperedness condition reduces to the motion of vortices and evolution of the field $\cphi(\vv z)$. 
%


%
%
\section{Dissipative Doi-Smoluchowski dynamics}

The generalized Doi-Smoluchowski kinetic equations \cite{DoiEdw} describe evolution of the density of liquid crystalline molecules $\odc(\vv s, \vv r; t)$ at a position $\vv r\in\R^3$ and orientation $\vv s \in \sph^2$. In general, the D-S dynamics includes the hydrodynamic interactions and diffusive transport of the spatial and orientational degrees of freedom. In this paper we consider the stage of evolution at which the hydrodynamic and diffusive transports of the spatial degrees of freedom have already equilibrated, and the evolution proceeds via diffusion of the orientational degrees of freedom. In this regime the concentration of liquid crystalline molecules is constant,
\begin{equation}
	c(\vv r;t)\ass\int_{\sph^2} \odc (\vv s, \vv r; t) \md \vv s\,=\,\const;
\end{equation}
and  the D-S equations may be rewritten in terms of the space-dependent probability density of molecules orientations, $\od(\vv s, \vv r; t)\ass\odc(\vv s, \vv r; t)/c(\vv r;t)$:
\begin{equation}\label{eq:doi2d1}
	\pd_t\od(\varphi,\vv z;t)\,=\,\pd_{\vv s}\cdot\bigg[\od\,\pd_{\vv s}\!\VD{\bf \eNrg}{\od}\bigg],
\end{equation}
where $\pd_{\vv s}$ and $\pd_{\vv s}\cdot$ denote the gradient and divergence operators on the sphere $\sph^2$; $\delta/\delta\od$ is the usual Euler-Lagrange variational derivative; and $\eNrg(\od)$ is the free energy of the system. In the two-dimensional model that we consider in this work, this equation becomes
\begin{equation}\label{eq:doi2d}
	\pd_t\od(\varphi,\vv z;t)\,=\,\pd_{\varphi}\bigg[\od\,\pd_{\varphi}\!\VD{\eNrg}{\od}\bigg],
\end{equation}
where $\varphi \in [0,2\pi)$, and $\eNrg (\od)$ is the spatially-extended Onsager-Maier-Saupe free energy \eqref{eq:fnrg2d}. Explicitly, equation (\ref{eq:doi2d}) may be written as
\begin{equation}\label{eq:dyn_doi_expl}
	\pd_t\od(\varphi,\vv z;t)\,=\,\pd^2_{\varphi\varphi}\od\,-\,
	2\im\big[\pd_\varphi(\od\me^{-2\mi\varphi})\opL\vv\dr\big],\qquad \opL\,=\,\eps^2\Delta+\gamma.
\end{equation}
Prescribing the boundary conditions on $\pd\dom$ for $\od$ directly is physically meaningless (there is no physical mechanism which would allow us to manipulate the density of orientations directly), and mathematically ill-posed. In this work we impose the Dirichlet boundary condition on the order parameter field: 
\begin{equation}
	\vv\dr(\vv z)\,=\,\dr_\eq\exp\big\{\mi\cpsi(\vv z)\big\},\qquad\vv z\in\pd\dom,
\end{equation}
where $\cpsi$ is defined on the boundary, possibly with a $2\pi k$-jump discontinuity somewhere on $\pd\dom$, as only values of $\cpsi(\vv z)\!\mod 2\pi$ are relevant.
Physically, this corresponds to the strong anchoring regime. Note, however, that Neumann, Robin, or mixed boundary conditions may be treated similarly, and result in different expressions for the renormalized multi-vortex energy.

We would like to study the dynamics prescribed by \eqref{eq:doi2d} in the limit when $\eps \ll 1$. This scaling corresponds to a regime when the defect cores shrink to a point and is motivated by our goal to understand the global evolution of patterns arising in the system, rather than the particular details of dynamics in vicinity of defect cores. Observe, that because of the energy decomposition \eqref{eq:red_nrg_decomp} in order to obtain a nontrivial dynamics, when the system does not immediately relax to the equilibrium state, we must consider \eqref{eq:doi2d} on a slower timescale $t'=\eps^2 t$ (dropping primes in what follows). The dynamics on this timescale is given by
\begin{equation}\label{eq:doi2d2}
	\pd_t\od(\varphi,\vv z;t)\,=\,\frac{1}{\eps^2} \pd_{\varphi}\bigg[\od\;\pd_{\varphi}\!\VD{\eNrg}{\od}\bigg].
\end{equation}

\paragraph{Summary of the results.}
In this work we show that, as $\eps\to 0$, the states of the system are close to multi-vortex configurations prescribed by (\ref{eq:th_lim_dir_gen}), and the dynamics \eqref{eq:doi2d2}  may be reformulated in terms of the dynamics of vortices $\vz{k}$ and the field $\cphi(\vv z)$. In particular, on the $\bO(1)$ timescale, the vortices are stationary, while the field $\cphi (z)$ evolves according to heat equation,
\begin{equation}\label{eq:phi_intro}
	{\pd_t}\cphi(\vv z;t)\,=\,\frac{4}{|\dom|\tau_\gamma}\Delta\cphi,\quad\vv z\in\dom,
\end{equation}
with the boundary condition
\begin{equation}\label{eq:bcond1}
	\cphi(\vv z;t)\,=\,\cpsi(\vv z)\,-\,\sum_{k=1}^N d_k\,\arg (\vv z - \vz{k}),\quad\vv z\in\pd\dom.
\end{equation}
Here $\tau_\gamma$ is related to parameters of the system via formula (\ref{eq:taugamma}). In order to obtain the motion of vortices, we must rescale the time yet again, introducing
\begin{equation}
	t^\prime\,=\,-\frac{8t}{\pi\tau_\gamma\ln\eps}.
\end{equation}
On this time scale, the function $\cphi(\vv z;t^\prime)$ is a harmonic function satisfying the same boundary condition (\ref{eq:bcond1}), while the vortices move according to the gradient flow equations generated by the renormalized multi-vortex energy given by \eqref{eq:th_nrg_vor_intro}:
\begin{align}\label{eq:z_intro}
	\vzd{k}(t^\prime)\,&=\,-\pd_{\vzb{k}}U\big(\vz{1},\ldots,\vz{N}\big).
\end{align}

We want to remark that in the natural regime of the Doi-Smoluchowski dynamics, which we consider in this paper, it is impossible to obtain a closed evolution equation for the order parameter field  $\vv\dr(\vv z)$, without immediately reducing it to the dynamics of vortices. This is a consequence of the fact that the reduction of the dynamics of $\od(\varphi,\vv z)$ to configurations defined by $\odn$ (which allows one to express all the relevant quantities in terms of $\vv\dr$) happens in the same limit as reduction of $\vv\dr(\vv z)$ to multi-vortex configurations prescribed by $\vv\dr_*$. This can be seen from equation (\ref{eq:nrg_decomp}), where the boundedness of the energy $\eNrg(\od)/\eps^2$ simultaneously imposes constraints on the relative entropy $\sNrg(\od|\odn)$ and the potential $W^\gamma(\dr)$, making the reduced free energy $\nNrg (\vv\dr)$ singular. Reduction to a theory involving exclusively the order parameter $\vv\dr$ would be possible if an additional large parameter appeared in front of the relative entropy term in  (\ref{eq:nrg_decomp}), without affecting the potential $W^\gamma$. However this is impossible due to the fact that both these terms appear as parts of the same entropic term in the Onsager energy \cite{Fat09a}. The situation is somewhat different in three dimensions, because in three-dimensional systems, the topological singularities do not possess infinite energy and thus the analogue of the reduced free energy $\nNrg(\vv\dr)$ remains nonsingular in such a limit. See \cite{E06} and \cite{Wang13}, where a reduction of Doi-Smoluchowski-type kinetic equations to Ericksen-Leslie equations was carried out.

In the following sections we will derive equations \eqref{eq:phi_intro} and \eqref{eq:z_intro} from the D-S evolution \eqref{eq:doi2d2} using ideas from the theory of gradient flows. To familiarize the reader with these ideas, we first discuss a finite dimensional example, in which we explain the methodology and derive equations for gradient flow dynamics constrained to a submanifold by a large drift generated  by the diverging part of the energy. Then we proceed to the analogous derivation for the D-S dynamics. In the final section of this paper, we rewrite equation \eqref{eq:doi2d2} in terms of an infinite hierarchy of evolution equations for moments of the orientation density and discuss possibilities for various closures of this hierarchy.

\section{A finite-dimensional example}
Suppose we are solving a (gradient flow) differential equation for $\vv x(t)\in\R^n$,
%
%
%
\begin{equation}\label{eq:grflow_ODE}
	\mybox{-0.5em}{1.5em}{\dot{\vv x}(t)\,=\,-D(\vv x)\,{\pd_{\vv x}}E(\vv x),}
\end{equation}
where $E:\R^n\to\R$ is the {\em energy} function; $D(\vv x)$ is a symmetric positive semi-definite $n\times n$ matrix. Assume that all the quantities that we employ are sufficiently regular to guarantee well-posedness of our formal derivations. Our first goal is to show that the vector problem (\ref{eq:grflow_ODE}) is equivalent to a single scalar inequality, which allows us to interpret solutions of (\ref{eq:grflow_ODE}) as {\em curves of maximal slope} for the energy function, $E$. Next, we will show, that if the energy depends on a small parameter, $\epsilon$, in such a way that, as $\epsilon\to0$, solutions become constrained to a submanifold of $\R^n$, we can describe the limiting curves as curves of maximal slope in some native parameterization of this manifold.
\paragraph{Formulation via curves of maximal slope.}
Consider an arbitrary curve $\vv x(t)$, $t\in[0,T]$, which is such that $\dot{\vv x}\in\rng D(\vv x)$ for all $t$. For variation of $E(\vv x)$ along this curve, we have
\begin{equation}
	E(t)-E(0)\,=\,\int_0^t{\pd_{\vv x}}E\big(\vv x(s)\big)\cdot{\dot{\vv x}}(s)\md s.
\end{equation}
Here we allow for a slight abuse of notation, employing $E(t)$ instead of $E\big(\vv x(t)\big)$. Let $G(\vv x)$ be the {\em generalized inverse} of $D(\vv x)$, in the sense that $G$ is a symmetric matrix with the same range and kernel as $D$, inverting $D$ in its range (for each $\vv x$). Such generalized inverse is defined uniquelly. Observe that both $D$ and $G$ are positive semi-definite, they commute, and posses unique symmetric square roots. Observe also that  $DG$ acts as identity on $\dot{\vv x}$. Thus we have, 
\begin{equation}
	E(0)-E(t)\,=\,-\int_0^t{\pd_{\vv x}}E\cdot\Big[\sqrt{DG}\,\Big]\,{\dot{\vv x}}(s)\md s\,=\,
	\int_0^t\Big[-\sqrt{D}\,{\pd_{\vv x}}E\Big]\cdot\Big[\sqrt{G}\,{\dot{\vv x}}(s)\Big]\md s.
\end{equation}
Using elementary inequalities, we obtain (omitting $s$-dependence),
\begin{equation}\label{eq:app_grd_flow_ineq1}
	E(0)-E(t)\stackrel{^{(\mathrm{a})}}{\leq}\,\int_0^t\big|\sqrt{D}\,{\pd_{\vv x}}E\big|\,\big|\sqrt{G}\,{\dot{\vv x}}\big|\md s\stackrel{^{(\mathrm{b})}}{\leq}\,
	\frac{1}{2}\int_0^t\left(\big|\sqrt{D}\,{\pd_{\vv x}}E\big|^2\,+\,\big|\sqrt{G}\,{\dot{\vv x}}\big|^2\right)\md s.
\end{equation}
Equality in (a) holds, if and only if $-\sqrt{D}\,{\pd_{\vv x}}E$ and $\sqrt{G}\,{\dot{\vv x}}$ are collinear. Equality in (b) holds, if and only if these quantities are equal by absolute values. Therefore, equalities in (\ref{eq:app_grd_flow_ineq1}) are attained, if and only if 
\begin{equation}
	\sqrt{G}\,{\dot{\vv x}}\,=\,-\sqrt{D}\,{\pd_{\vv x}}E(\vv x).
\end{equation}
Multiplying both sides by $\sqrt{D}$, we recover (\ref{eq:grflow_ODE}). Thus, by reversing the inequalities in (\ref{eq:app_grd_flow_ineq1}), we obtain an inequality which is equivalent to the differential equation (\ref{eq:grflow_ODE}):
\begin{equation}\label{eq:app_grd_flow_ineq2}
	\mybox{-0.5em}{1.5em}{
	E(0)-E(t)\,\geq\,\frac{1}{2}\int_0^t\left(\big\|{\pd_{\vv x}}E\big(\vv x(s)\big)\big\|_D^2\,+\,\big\|{\dot{\vv x}(s)}\big\|_G^2\right)\md s.}
\end{equation}
Here we denoted (explicitly writing out the derivatives and employing Einstein summation rules)
\begin{equation}\label{eq:curvemaxslope}
	\big\|{\pd_{\vv x}}E(\vv x)\big\|_D^2\,\ass\,D^{ij}(\vv x)\,{\pd_j}E(\vv x)\,{\pd_i}E(\vv x);\qquad
	\big\|{\dot{\vv x}}\big\|_G^2\,\ass\,G_{ij}(\vv x)\,\dot x^i\dot x^j.
\end{equation}
We say that a curve $\vv x(t)$ is a {\em curve of maximal slope} for the energy function $E(\vv x)$ in metric prescribed by $G(\vv x)$, if $\dot{\vv x}\in\rng G(\vv x)$, and the inequality (\ref{eq:app_grd_flow_ineq2}) holds for (almost) all $t$. 

In this derivation we started from a differential equation and obtained a scalar inequality, i.e., we showed that solutions of (\ref{eq:grflow_ODE}) are curves of maximal slope for $E(\vv x)$, and vice versa. One could, however, start with (\ref{eq:app_grd_flow_ineq2}) in a rather general metric space setting, and prove that the curves of maximal slope exist, posses certain regularity properties, and satisfy some differential equations, whenever the energy and the metric are sufficiently regular themselves. Such developments may be found in the book by Ambrosio, Gigli, and Savar\'e \cite{Ambro05}. Whenever we discuss gradient flow equations in this work, we understand them in terms of the curves of maximal slope formulation.

\paragraph{Change of variables.}
Suppose we want to study a family of curves of maximal slope which lie in an $m$-dimensional submanifold, $\mathcal M$, of $\R^n$ in a parameterization native to $\mathcal M$. In other words, we assume that $\vv x(t)=\vv\chi(\vv y(t))$ for $\vv y\in\R^m$, $m\leq n$, and some map $\vv{\chi}:\R^m\to\R^n$; and we want to obtain description of our curves using the $\vv y$-variables. We will always work within the same chart of $\mathcal M$, and will  not worry about chart transitions here.

Using the chain rule, we immediately get, (employing Greek indices for the $\vv y$-variables)
\begin{equation}\label{eq:chvar_ineq1}
	\big\|{\dot{\vv x}}\big\|_G^2\,=\,G_{ij}(\vv x)\,\dot x^i\dot x^j\,=\,G_{ij}\big(\vv\chi(\vv y)\big)\,\Big[{\pd_\alpha}\chi^i(\vv y)\dot y^\alpha\Big]\Big[{\pd_\beta}\chi^j(\vv y)\dot y^\beta\Big]\ass\,
	\tilde G_{\alpha\beta}(\vv y)\,\dot y^\alpha\dot y^\beta\,=\,\big\|{\dot{\vv y}}\big\|_{\tilde G}^2,
\end{equation}
where the $m\times m$ matrix $\tilde G(\vv y)$ is defined as
\begin{equation}\label{eq:grd_red1}
	{\tilde G}_{\alpha\beta}(\vv y)\,=\,{\pd_\alpha}\chi^i(\vv y)\,G_{ij}\big(\vv\chi(\vv y)\big)\,{\pd_{\beta}}\chi^j({\vv y}).
\end{equation}
The matrix $\tilde G(\vv y)$ has a simple geometric interpretation: it is the metric induced by $G$ on $\mathcal M$, expressed in the $\vv y$-parameterization.  Define $\tilde D(\vv y)$ as the generalized inverse of $\tilde G(\vv y)$; and set
\begin{equation}\label{eq:grd_red}
	\tilde E(\vv y)\,\ass\,E\big(\vv\chi(\vv y)\big).
\end{equation}
Let us make a few additional assumptions, which are not required, but simplify some of the following arguments. Suppose there exists a neighborhood of $\mathcal M$, in which there exists a non-degenerate map $\vv\eta:\R^n\to\R^m$, such that 
\begin{equation}
	\vv y\,=\,\vv\eta\big(\vv\chi(\vv y)\big).
\end{equation}
Assume also that the exists another non-degenerate map $\vv\zeta:\R^n\to\R^{n-m}$, such that $\mathcal M$ is the 0-level set of $\vv\zeta$, and ${\pd_{\vv x}}\vv\zeta$ is orthogonal to ${\pd_{\vv x}}\vv\eta$, i.e.,
\begin{equation}
	{\pd_i}\zeta^\nu(\vv x)\, D^{ij}(\vv x)\, {\pd_j}\eta^\alpha(\vv x)\,=\,0;\qquad\alpha=1\ldots m;\quad\nu=1\ldots n-m.
\end{equation}
Decomposing the $\vv x$-gradient of $E$ into the sum of gradients with respect to $\vv \eta$ and $\vv\zeta$, we get
\begin{align}\label{eq:grad_norm}
	\big\|{\pd_{\vv x}}E\big\|^2_D\,&=\,\big\|{\pd_{\vv\eta}}E\,{\pd_{\vv x}}\vv\eta\,+\,{\pd_{\vv\zeta}}E\,{\pd_{\vv x}}\vv\zeta\big\|^2_D\,=\,
	\big\|{\pd_{\vv\eta}}E\,{\pd_{\vv x}}\vv\eta\big\|^2_D\,+\,{\pd_\nu}E\Big[{\pd_i}\zeta^\nu\, D^{ij}\, {\pd_j}\eta^\alpha\Big]{\pd_\alpha}E\,+\,
	\big\|{\pd_{\vv\zeta}}E\,{\pd_{\vv x}}\vv\zeta\big\|^2_D\nonumber\\
	&=\,\big\|{\pd_{\vv\eta}}E\,{\pd_{\vv x}}\vv\eta\big\|^2_D\,+\,\big\|{\pd_{\vv\zeta}}E\,{\pd_{\vv x}}\vv\zeta\big\|^2_D\,\geq\,
	\big\|{\pd_{\vv\eta}}E\,{\pd_{\vv x}}\vv\eta\big\|^2_D.
\end{align}
When $\vv x=\vv\eta(\vv y)\in\mathcal M$,  the last term in (\ref{eq:grad_norm}) is exactly $\big\|{\pd_{\vv y}}\tilde E\big\|^2_{\tilde D}$. Thus we get
\begin{equation}\label{eq:chvar_ineq2}
	\big\|{\pd_{\vv x}}E\big(\vv\chi(\vv y)\big)\big\|_D^2\,\geq\,\big\|{\pd_{\vv y}}\tilde E(\vv y)\big\|_{\tilde D}^2.
\end{equation}
This inequality expresses the fact that by extending $\tilde E$ as $E$ from $\mathcal M$ into $\R^n$, we can only increase the norm of its gradient. Combining equations (\ref{eq:chvar_ineq1}), (\ref{eq:chvar_ineq2}), and (\ref{eq:app_grd_flow_ineq2}), we get
\begin{equation}\label{eq:app_grd_flow_chvar}
	\tilde E(0)-\tilde E(t)\,\geq\,\frac{1}{2}\int_0^t\left(\big\|{\pd_{\vv y}}\tilde E\big(\vv y(s)\big)\big\|_{\tilde D}^2\,+\,\big\|{\dot{\vv y}(s)}\big\|_{\tilde G}^2\right)\md s.
\end{equation}
Thus we demonstrated that curves of maximal slope in $\vv x$-parameterization of $\R^n$ are also curves of maximal slope in $\vv y$-parameterization of $\R^m$, when the latter is equipped with metric inherited through its embedding as the submanifold of $\R^n$.

\paragraph{Asymptotic reduction.} Suppose now, our energy depends on a small parameter, $\epsilon$, and the dynamics is such that, as $\epsilon\to0$, all the trajectories $\vv x^\epsilon(t)$ become constrained onto $\mathcal M$. Let us derive equations which describe this asymptotic dynamics in terms of the $\vv y$-variables. 

Consider the energy function of the following form:
\begin{equation}
	E^\epsilon(\vv x)\,=\,U(\vv x)\,+\,\frac{1}{\epsilon}V(\vv \zeta(\vv x)),
\end{equation}
where $\vv\zeta(\vv x)$ is as above.  Assume that $U:\R^n\to\R$ is bounded below; $V:\R^{n-m}\to\R^+$ has minimum at the origin and has no other critical points; without loss of generality, set $V(0)=0$. This construction is designed so, that the ``fast'' flow generated by $V$ will quickly carry solutions to the vicinity of $\mathcal M$, while it will not affect the dynamics on $\mathcal M$ itself.

Pick a sequence of initial conditions, such that
\begin{subequations}
\begin{align}
	\vv x^\epsilon(0)&\to\vv x^0(0)\in\mathcal M;\\\label{eq:conv_nrg}
	E^\epsilon\big(\vv x^\epsilon(0)\big)\,&\to\,U\big(\vv x^0(0)\big).
\end{align}
\end{subequations}
The second condition assures that there is no excess energy in the system. Generally, one can show (by other methods) that this condition is not required, as it will be automatically satisfied after some initial time of $\scriptstyle{\mathcal O}\displaystyle(1)$. Assume that, as $\epsilon\to0$,
\begin{equation}
	\vv x^\epsilon(t)\to\vv x^0(t),\quad\text{pointwise for}~t\in[0,T].\\
\end{equation}
This may be proven for sufficiently regular $U$ and $V$. The energy is non-increasing along the curves of maximal slope, therefore $V\big(\vv x^\epsilon(t)\big)$ must remain of $\bO(\epsilon)$ for all $t>0$. This implies that $\smash{\vv x^0(t)\in\mathcal M}$ for all $t\geq0$.

We will now show that $\vv x^0(t)$ is a curve of maximal slope for $U$ on $\mathcal M$ equipped with metric inherited from $\R^n$. First of all, observe, that due to positivity of $V$,
\begin{equation}
	E^\epsilon\big(\vv x^\epsilon(0)\big)-E^\epsilon\big(\vv x^\epsilon(t)\big)\,\leq\,E^\epsilon\big(\vv x^\epsilon(0)\big)-U\big(\vv x^\epsilon(t)\big).
\end{equation}
Passing to the limit as $\epsilon\to0$, using (\ref{eq:conv_nrg}) and the continuity of $U$, we get
\begin{equation}\label{eq:asred3}
	\lim_{\epsilon\to0}\Big[E^\epsilon\big(\vv x^\epsilon(0)\big)-E^\epsilon\big(\vv x^\epsilon(t)\big)\Big]\,\leq\,U\big(\vv x^0(0)\big)-U\big(\vv x^0(t)\big)\,=\,
	\tilde U\big(\vv y(0)\big)-\tilde U\big(\vv y(t)\big).
\end{equation}
The pointwise convergence of $\vv x^\epsilon(t)$ to $\vv x^0(t)$ implies
\begin{equation}\label{eq:asred2}
	\liminf_{\epsilon\to0}\int_0^t\big\|\dot{\vv x}^\epsilon(s)\big\|_{G(\vv x^\epsilon)}^2\md s\geq
	\int_0^t\big\|\dot{\vv x}^0(s)\big\|_{G(\vv x^0)}^2\md s\,=\,\int_0^t\big\|\dot{\vv y}(s)\big\|_{\tilde G(\vv y)}^2\md s.
\end{equation}
From (\ref{eq:grad_norm}), using that $V$ only depends on $\vv \zeta$, we get that
\begin{align}
	\big\|{\pd_{\vv x}}E^\epsilon(\vv x^\epsilon)\big\|^2_{D(\vv x^\epsilon)}\,\geq\,
	\big\|{\pd_{\vv\eta}}E^\epsilon(\vv x^\epsilon)\,{\pd_{\vv x}}\vv\eta(\vv x^\epsilon)\big\|^2_{D(\vv x^\epsilon)}\,=\,\big\|{\pd_{\vv\eta}}U(\vv x^\epsilon)\,{\pd_{\vv x}}\vv\eta(\vv x^\epsilon)\big\|^2_{D(\vv x^\epsilon)}.
\end{align}
Therefore
\begin{equation}\label{eq:asred1}
	\liminf_{\epsilon\to0}\int_0^t\big\|{\pd_{\vv x}}E^\epsilon\big(\vv x^\epsilon\big)\big\|^2_{D(\vv x^\epsilon)}\md s\,\geq\,
	\int_0^t\big\|{\pd_{\vv\eta}}U\big(\vv x^0\big)\,{\pd_{\vv x}}\vv\eta\big(\vv x^0\big)\big\|^2_{D(\vv x^0)}\md s\,=\,
	\int_0^t\big\|{\pd_{\vv y}}\tilde U(\vv y)\big\|^2_{\tilde D(\vv y)}\md s.
\end{equation}
Using inequalities (\ref{eq:asred3}), (\ref{eq:asred2}), and (\ref{eq:asred1}) in (\ref{eq:app_grd_flow_ineq2}), we get the desired result:
\begin{equation}\label{eq:app_grd_flow_asy}
	\tilde U(0)-\tilde U(t)\,\geq\,\frac{1}{2}\int_0^t\left(\big\|{\pd_{\vv y}}\tilde U\big(\vv y(s)\big)\big\|_{\tilde D}^2\,+\,\big\|{\dot{\vv y}(s)}\big\|_{\tilde G}^2\right)\md s.
\end{equation}
Thus we see that the limiting trajectories are curves of maximal slope for the ``slow'' part of the energy, $U$, constrained to $M$. Using the equivalence of this formulation to formulation via gradient differential equations, we can also state this result in the following manner: the limiting trajectories may be obtained as $\vv x^0(t)\,=\,\vv \chi(\vv y(t))$, where $\vv y(0)=\vv\eta(\vv x(0))$ and $\vv y(t)$ satisfies
\begin{equation}
	\dot{\vv y}(t)\,=\,-\tilde D(\vv y)\,{\pd_{\vv y}}\tilde U(\vv y).
\end{equation}
Note that the matrix $\tilde D(\vv y)$ must be computed by inverting the matrix of the metric tensor $\tilde G(\vv y)$ given by (\ref{eq:grd_red1}). Calculation of these matrices becomes the only ingredient required for obtaining the limiting dynamics.


\section{Derivation for the Doi-Smoluchowski dynamics}\label{sec:Doi}
The kinetic equation (\ref{eq:doi2d}) formally resembles our finite-dimensional ODE example. The density of orientations $\od(\varphi,\vv z)$ plays the role of $\vv x$-variables, while the vortex locations $\vv z_k$ and the function $\cphi(\vv z)$ correspond to the reduced $\vv y$-variables. Thus we will proceed along the same lines in this derivation.

 \paragraph{Mobility and metric.} As in the finite-dimensional example, we can obtain that the dynamics (\ref{eq:doi2d}) is equivalent to the following inequality:
\begin{equation}
	\eNrg (0) - \eNrg(t) \,\geq \,\int_0^t\, \left(  \left\|\VD{\eNrg}{\od}\right\|^2_{\hat D}  \,+ \, \big\| \pd_t \od \big\|^2_{\hat G} \right) \, \md t ,
\end{equation}
where the operator $\hat D[\od]f(\varphi)\ass-{\pd_\varphi}\big(\od\,{\pd_\varphi}f(\varphi)\big)$ corresponds to  the matrix $D(\vv x)$ in \eqref{eq:grflow_ODE}, and $\smash{\hat G = \hat D^{-1}}$ is the generalized inverse of $\hat D$. In order to determine $\hat G$, we must solve the differential equation,
\begin{equation}\label{eq:Geq1}
	-{\pd_\varphi}\big(\od(\varphi)\,{\pd_\varphi}f(\varphi)\big)\,=\,u(\varphi).
\end{equation}
Let us assume that the support of $\od(\varphi)$ is the entire interval, $[0,2\pi)$, and treat $\smash{\hat D}$ and $\smash{\hat G}$ as symmetric operators defined on smooth $2\pi$-periodic functions. Integrating (\ref{eq:Geq1}) once, we get
\begin{equation}\label{eq:transport1}
	v(\varphi)\,\ass\,{\pd_\varphi} f(\varphi)\,=\,-\frac{1}{\od(\varphi)}\int_0^{\varphi} u(\varphi^\prime)\md\varphi^\prime\,+\,\frac{C}{\od(\varphi)}.
\end{equation}
%
%
The function $v(\varphi)$ is a derivative, and thus its total integral must vanish; this gives us the condition,
\begin{equation}\label{eq:transport2}
	C\,=\,\bigg[\int_0^{2\pi}\frac{\md\varphi}{\od(\varphi)}\bigg]^{-1}\int_0^{2\pi}\frac{1}{\od(\varphi)}\int_0^{\varphi} u(\varphi^\prime)\md\varphi^\prime\md\varphi.
\end{equation}
We do not need to integrate equation (\ref{eq:Geq1}) second time, as it is convenient to define $\hat G$ using $v(\varphi)$. We only need $\hat G$ as a bilinear form; for its action we have, whenever $u,\tilde u,f,\tilde f\in\rng\hat G$,
\begin{equation}
	(u,\hat G\tilde u)\,=\,(\hat D f,\hat G\hat D\tilde f)\,=\,(\hat D f, \tilde f)\,=\,(-\pd_\varphi \od\pd_\varphi f,\tilde f)\,=\,(\od\,{\pd_\varphi}f,{\pd_\varphi}\tilde f).
\end{equation}
Writing this down explicitly in terms of $v(\varphi)$ and $\tilde v(\varphi)$,
\begin{equation}\label{eq:w2prod}
	(u,\hat G\tilde u)\,=\,\big(u,\big[-\pd_\varphi\od\pd_\varphi\big]^{-1}\tilde u\big)\,\ass\,
	\int_\dom\int_0^{2\pi}v(\varphi,\vv z)\,\tilde v(\varphi,\vv z)\;\od(\varphi,\vv z)\md\varphi\dvz.
\end{equation}

\paragraph{Structure of the slow manifold.}
The manifold $\mathcal M$ corresponds to the set of optimal orientation densities produced by multi-vortex maps. Let us use the hat symbol, ``$\hat{~}$'' to denote such configurations; we have, as in (\ref{eq:fam2d}),
\begin{equation}\label{eq:fam2dlast}
	\od(\varphi,\vv z)\,=\,\odnh(\varphi,\vv z)\,=\,\Big[2\pi\ibes_0\!\big(\!\abes(\hat\dr)\big)\Big]^{-1}{\exp\Big[\abes(\hat\dr)\cos(2\varphi-\arg\hat{\vv\dr})\Big]},
\end{equation}
where the order parameter field is parameterized by vortex locations, $\vz{k}$, $k=1,\ldots,N$, and the phase function $\cphi(\vv z)$:
\begin{equation}
	\quad \hat{\vv\dr}[\vz{1},\ldots,\vz{k};\cphi](\vv z)\,=\,\hat\dr[\vz{1},\dots,\vz{N}](\vv z)\;\exp\Big[\mi\cphi(\vv z)\,+\,\mi\sum_k \;d_k\arg(\vv z-\vz{k})\Big].
\end{equation}
Even though we do not know the exact shape of the function $\hat\dr[\vz{1},\dots,\vz{N}](\vv z)$, the finiteness of the energy $\nNrg(\vv\dr)$ implies that $\hat\dr(\vv z)$  turns to zero at the location of vortices and approaches $\dr_\eq$ at distances larger than $\bO(\eps)$. However, the specifics of this behavior are not important for our purposes. It is sufficient to utilize the following property, which is well-known in the context of the Ginzburg-Landau theory \cite{San98}. Given a tempered family of order parameter fields converging to a multi-vortex configuration, there exists a covering of the vortices by disks $\ball{\vz{k}}{R_\epsilon}$ with radii $R_\eps=\bO(\eps)$, such that $|\dr^\epsilon(\vv z)-\dr_\eq|=\sO(1)$ in the exterior of these disks, while
\begin{equation}\label{eq:covering}
	\int_{\ball{\vz{k}}{R_\epsilon}} |\nabla\vv\dr^\epsilon(\vv z)|^2\dvz \,\leq \,C.
\end{equation}
This bound will be used below to estimate the matrix elements of $\hat G$ restricted to multi-vortex configurations.
\paragraph{Change of variables.}
In our liquid crystalline system, the role of the map $\vv\chi$ is played by $\odh$; variables parameterizing the slow manifold, $\mathcal M$, are the vortex locations, $\vz{k}$, and the function $\cphi(\vv z)$. Let us define $\vv Y\ass(\vz{1}, \ldots, \vz{N}, \Phi)$. Let the index $\alpha$ run through these parameters. Using expression \eqref{eq:red_nrg_decomp} for the decomposition of energy $\eNrg (\od)$, and proceeding formally in the same way as in the  finite dimensional example, we obtain,
\begin{equation}
	\tilde \nNrg (0)\, -\, \tilde \nNrg(t) \,\geq\, \int_0^t \left(  \left\| \VD{\tilde \nNrg}{\vv Y}\right \|^2_{D} + \big\| {\pd_t}\vv Y \big\|^2_{G} \right) \, \md t.
\end{equation}
This inequality is equivalent to the following differential equations for the reduced dynamics:
\begin{equation}
	{\pd_t}\vv Y = - D\,\VD{\tilde \nNrg}{\vv Y}.
\end{equation}
Thus we need to compute the matrix $G$ and its generalized inverse $D = G^{-1}$. We start by computing the analogues of the derivatives $\pd_{\alpha}\chi^{i}(\vv y)$.  The chain rule gives us,
\begin{equation}
	{\pd_\alpha}\odh\,=\,\Big[{\pd_{\dr}}\odh\Big]\,{\pd_\alpha}\hat\dr\,+\,\Big[{\pd_{\arg\vv\dr}}\odh\Big]\,{\pd_\alpha}\arg\hat{\vv\dr}.
\end{equation}
Observing that ${\pd_{\arg\vv\dr}\odh}\,=\,-\pd_\varphi\odh/2$, we can write
\begin{subequations}
\begin{align}
	{\pd_{\cphi}}\odh\,&=\,\mybox{0em}{0em}{-\frac{1}{2}\pd_\varphi\odh;}^{~\text{(a)}}\\\label{eq:Gcrossterms}
	{\pd_{\vz{k}}}\odh\,&=\,\mybox{0em}{0em}{-\frac{\mi d_k}{4(\vv z-\vz{k})}\pd_\varphi\odh}^{~\text{(b)}}\!+~\mybox{-0.95em}{2.4em}{{\pd_\dr}\odh\,{\pd_{\vz{k}}}\hat\dr.}^{~\text{(c)}}
\end{align}
\end{subequations}
In order to compute the matrix
\begin{equation}\label{eq:finalG}
	G_{\alpha\beta}\,\ass\,({\pd_\alpha}\odh,\hat G\,{\pd_\beta}\odh)\,=\,\int_\dom\int_0^{2\pi}v_\alpha(\varphi,\vv z)\,v_\beta(\varphi,\vv z)\;\od(\varphi,\vv z)\md\varphi\dvz,
\end{equation}
we must find the fields $v_\alpha(\varphi,\vv z)$ by solving the differential equations, 
\begin{equation}\label{eq:tr_chvar}
	-\pd_\varphi(\odh\,v_\alpha)\,=\,\pd_\alpha\odh,
\end{equation}
as described in formulas (\ref{eq:transport1}) and (\ref{eq:transport2}), and after that, evaluate the integrals in (\ref{eq:finalG}). Let us implement this plan for $G_{\vz{k}\vzb{k}}$. As equation (\ref{eq:tr_chvar}) is linear, its solution may be represented as a sum of solutions with right-hand sides corresponding to terms labelled as (b) and (c) in formula (\ref{eq:Gcrossterms}).  For the solution corresponding to term (b), omitting the $\vv z$-dependence, we get
\begin{equation}
	v(\varphi)\,=\,\frac{2\pi}{\odh(\varphi)}\bigg[\int_0^{2\pi}\frac{\md\varphi^\prime}{\odh(\varphi^\prime)}\bigg]^{-1}-\,1\,=\,
	\frac{1}{2\pi\Ibes_0^2(\Abes\!\big(\hat\dr)\big)\odh(\varphi)}\,-\,1.
\end{equation}
For term (c), we first compute
\begin{equation}\label{eq:rhs2}
	\Pd{\odh}{\dr}\,=\,\Abes^\prime(\dr)\big[\cos(2\varphi-\arg\vv\dr)\,-\,\dr\big]\odh(\varphi);
\end{equation}
now we can see that the solution to (\ref{eq:tr_chvar}) with right-hand side given by (\ref{eq:rhs2}) may be represented as $\Abes^\prime(\dr)F(\varphi,\arg\vv\dr)$, where $F$ is some bounded function. Thus we have
\begin{equation}\label{eq:vzk_exact}
	v_{\vz{k}}(\varphi)\,=\,-\frac{\mi d_k}{4(\vv z-\vz{k})}\bigg[\frac{1}{2\pi\Ibes_0^2(\Abes\!\big(\hat\dr)\big)\odh(\varphi)}\,-\,1\bigg]\,+\,\Abes^\prime(\dr)F(\varphi,\arg\vv\dr)\,{\pd_{\vz{k}}}\hat\dr.
\end{equation}
For $v_{\vzb{k}}$ we get the expression, complex-conjugate to (\ref{eq:vzk_exact}). Now we compute the integrals in (\ref{eq:finalG}). As, the only non-integrable singularity which appears in the calculation is $\smash{1/|\vv z|^2}$, all the terms, except the product of the first term in (\ref{eq:vzk_exact}) and its complex conjugate, contribute in $\bO(1)$ as $\eps\to0$. Integrating over $\varphi$, we obtain
\begin{equation}\label{eq:int1}
	G_{\vz{k}\vzb{k}}\,=\,\frac{1}{16}\int_\dom\bigg[1\,-\,\frac{1}{\Ibes^2_0\!\big(\!\Abes(\hat\dr(\vv z))\big)}\bigg]\frac{\dvz}{|\vv z-\vz{k}|^2}\,+\,\bO(1).
\end{equation}
In order to estimate this integral, we first split the domain $\dom$ in two parts: $\ball{R_\epsilon}{z_k}$ and $\dom \setminus \ball{R_\epsilon}{z_k}$, where the radius $R_\epsilon=\bO(\epsilon)$ is chosen so that the inequality (\ref{eq:covering}) holds.  Using the properties of the function $\Abes(\dr)$ (see formula (\ref{eq:app_ineqrad}) in the appendix), and taking into account (\ref{eq:covering}), we obtain
\begin{equation}
	G_{\vz{k}\vzb{k}} \leq C_1\int_{\ball{R_\epsilon}{z_k}} |\nabla \vv\dr|^2\dvz \leq C_2.
\end{equation}
For the integral over the exterior of the disk $\ball{R_\epsilon}{z_k}$ we can use Lemma~\ref{lemma1} from the appendix, obtaining in the end,
\begin{equation}
	G_{\vz{k}\vzb{k}}\,=\,-\frac{\pi}{8}\left[1\,-\,\frac{1}{\Ibes^2_0\!\big(\!\Abes(\dr_\eq)\big)}\right]\ln\epsilon\,+\,\bO(1).
\end{equation}
The calculation of $G_{\cphi\cphi}$ is equivalent to this one. We do not provide detailed calculations for other matrix elements of $G$ here; it is only important to note that they all are of $\bO(1)$ as $\eps\to0$. Summarizing all these calculations we have,
\begin{subequations}
\begin{align}
	G_{\vz{k}\vzb{k}}\,&=\,
	-\frac{\pi\tau_\gamma}{8}\ln\eps\,+\,\bO(1),\quad k=1,\ldots,N;
	\\
	G_{\cphi\cphi}\,&=\,\frac{|\dom|\tau_\gamma}{4}\,+\,\bO(\eps);
	\\
	G_{\alpha\beta}\,&=\,\bO(1),\quad\text{for all other combinations of~}\alpha\text{~and~}\beta.\rule{0em}{1.5em}
\end{align}
\end{subequations}
Here we introduced a scaling factor,
\begin{equation}\label{eq:taugamma}
	\tau_\gamma\,\ass\,1\,-\,\frac{1}{\Ibes^2_0\!\big(\!\Abes(\dr_\eq)\big)}\,=\,1\,-\,\frac{1}{\Ibes^2_0(\gamma\dr_\eq)};
\end{equation}
for the second equality we used that $\dr_\eq$ satisfies $\gamma\dr_\eq\,=\,\Abes(\dr_\eq)$.

The final step in our calculation is computation of $D$, the inverse of $G$. Using Lemma~\ref{lemma} from the appendix (setting $\delta=-1/\ln\eps$), we obtain
\begin{subequations}
\begin{align}
	D_{\vz{k}\vzb{k}}\,&=\,
	-\frac{8}{\pi\tau_\gamma\ln\eps}\,+\,\bO(1/\ln^2\eps),\quad
	D_{\vz{k}\vz{k}}\,=\,\bO(1/\ln^2\eps),
	\quad k=1,\ldots,N;
	\\
	D_{\cphi\cphi}\,&=\,\frac{4}{|\dom|\tau_\gamma}\,+\,\bO(1/\ln\eps);
	\\
	D_{\alpha\beta}\,&=\,\bO(1/\ln\eps),\quad\text{for all other combinations of~}\alpha\text{~and~}\beta.\rule{0em}{1.25em}
\end{align}
\end{subequations}

\paragraph{Evolution equations for $\vz{k}$, $\cphi$, and timescales.}
Now, once we have computed the matrix $D$, we are finally able to write down equations governing the evolution of vortices, $\vz{k}$, and the function $\cphi(\vv z)$ in the limit as $\eps\to0$:
\begin{subequations}
\begin{align}\label{eq:vort_last}
	\vzd{k}\,&=\,\frac{8}{\pi\tau_\gamma\ln\eps}\pd_{\vzb{k}}\tilde E(\vz{1},\ldots,\vz{N};\cphi)\,+\,\bO(1/\ln\eps)\VD{\tilde E}{\cphi}\,+\,\bO(1/\ln^2\eps);\\
	{\pd_t}\cphi\,&=\,-\frac{4}{|\dom|\tau_\gamma}\VD{\tilde E}{\cphi}\,+\,\bO(1/\ln\eps)\,=\,\frac{4}{|\dom|\tau_\gamma}\,\Delta\cphi\,+\,\bO(1/\ln\eps).
\end{align}
\end{subequations}
From these equations we can see that, as $\eps\to0$,  the vortices are stationary, and the only evolution which occurs in our system is the relaxation of the field $\cphi(\vv z)$ governed by the heat equation,
\begin{equation}
	{\pd_t}\cphi(\vv z;t)\,=\,\frac{4}{|\dom|\tau_\gamma}\Delta\cphi,\quad\vv z\in\dom.
\end{equation}
The boundary condition is inherited from our Dirichlet boundary condition on $\vv\dr$:
\begin{equation}\label{eq:bcond}
	\cphi(\vv z;t)\,=\,\cpsi(\vv z)\,-\,\sum_{k=1}^N d_k\,\arg (\vv z - \vz{k}),\quad\vv z\in\pd\dom.
\end{equation}
In order to capture the vortex dynamics, we introduce a slow, rescaled time,
\begin{equation}
	t^\prime\,=\,-\frac{8t}{\pi\tau_\gamma\ln\eps}.
\end{equation}
On this time scale, in the leading order, $\cphi$ is a harmonic function satisfying the same boundary condition (\ref{eq:bcond}). We get that on the $t^\prime$-timescale, the evolution of vortices satisfies
\begin{align}
	\vzd{k}(t^\prime)\,&=\,-\pd_{\vzb{k}}U\big(\vz{1},\ldots,\vz{N}\big).
\end{align}
Note, that the second term on the right-hand side of equation (\ref{eq:vort_last}) vanishes because on the $t^\prime$-timescale, $\delta\tilde E/\delta\cphi=0$.

\subsection{Equations for the moments and closures}
For the sake of completeness, here we derive equations for the moments of the orientation density (its Fourier coefficients) and formally perform a closure at the level of the first moment. This closure, even though sensible from the physical standpoint, does not have a valid mathematical justification, and does not occur in some well-defined asymptotic limit.

\paragraph{Equations for the moments.}
Let us define $k$-th moment (Fourier coefficient) of the orientations density $\od(\varphi,\vv z;t)$ as
\begin{equation}\label{eq:nk}
	\vv\dr^{(k)}(\vv z;t)\,\ass\,\int_0^{2\pi}\me^{\,2\mi k\varphi}\od(\varphi,\vv z;t)\md\varphi.
\end{equation}
The factor of 2 in $\me^{2\mi k\varphi}$ appears because, physically, in nematic systems, the orientations density is invariant with respect to inversion of liquid crystalline molecules, and thus $\od(\varphi)=\od(\varphi+\pi)$, and all the odd Fourier coefficients of the orientation density vanish. The first moment $\smash{\vv\dr^{(1)}}$ is exactly the order parameter field, $\vv\dr$, employed in our work. 

In order to obtain dynamic equations for $k$-th moment we can differentiate equation (\ref{eq:nk}) with respect to time $t$ and use evolution equation (\ref{eq:dyn_doi_expl}) for $\od$, obtaining
\begin{equation}\label{eq:moments}
	\pd_t{\vv{\dr}}^{(k)}(\vv z;t)\,=\,-\,4k^2\,\vv{\dr}^{(k)}\,+\,2k\,\left[\vv{\dr}^{(k-1)}\opL\vv\dr\,-\,\vv{\dr}^{(k+1)}\opL\bar{\vv\dr}\right].
\end{equation}
It is possible to rewrite these equations in a gradient form using the energy decomposition \eqref{eq:nrg_decomp}:
\begin{equation}\label{eq:momentsgrd}
	\pd_t\vv{\dr}^{(k)}(\vv z;t)\,=\,4k\,\left[\vv\dr^{(k+1)}\,\VD{\nNrg}{\vv\dr}\,-\,\vv\dr^{(k-1)}\,\VD{\nNrg}{\bar{\vv\dr}}\right]\,
	-\,2\mi k\int_0^{2\pi}\me^{\,2\mi k\varphi}\,\left[\VD{\sNrg(\od|\odn)}{\od}\right]_\varphi\,\od(\vv z,\varphi;t)\md\varphi
\end{equation}
Equations \eqref{eq:moments} or \eqref{eq:momentsgrd} form an infinite hierarchy, as equation for each $\vv\dr^{(k)}$ involves $\vv\dr^{(k+1)}$, etc. In order to obtain a closed system of equations for some first few moments, one must find a way to decouple this hierarchy by expressing the higher-order moments via the lower-order ones. This requires some additional assumptions on the orientation density $\od$.

\paragraph{Maximal entropy closure.}
A natural physical assumption is that the orientation density relaxes to its optimal configuration, $\odn$, given by \eqref{eq:fam2d}, which minimizes the relative entropy term in the energy (\ref{eq:nrg_decomp}). (Physical entropy is defined with a sign opposite to the one used here, and thus the name, ``maximal entropy.'') 
This allows us to calculate the higher-order moments in terms of $\vv\dr$ explicitly: 
\begin{equation}
	\vv\dr^{(k)}(\vv z)\,=\,\int_0^{2\pi}\me^{2\mi k\varphi}\od(\varphi,\vv z)\md\varphi\,=\,\frac{\Ibes_k(\Abes(\dr))}{\Ibes_0(\Abes(\dr))}
	\me^{\,\mi k\arg\vv\dr(\vv z)}.
\end{equation}
In particular, we find that 
\begin{equation}
	\vv\dr^{(2)}(\vv z)\,= \, \frac{\Ibes_2(\Abes(\dr))}{\Ibes_0(\Abes(\dr))} \me^{\,\mi 2\arg\vv\dr(\vv z)} \, =\, \frac{{\vv\dr}^2}{\dr^2} \left[ 1- \-\frac{2\dr}{\Abes(\dr)} \right].
\end{equation}
Now it is possible to close the hierarchy \eqref{eq:moments} at the level of the first moment:
\begin{align}\label{eq:nclosed}
	\pd_t{\vv{\dr}}(\vv z;t)\,
	&=\,-\,4\vv{\dr}\,+\,2\,\left[\opL\vv\dr\,-\,\frac{{\vv\dr}^2}{\dr^2}\left(1-\frac{2\dr}{\Abes(\dr)}\right)
	\opL\bar{\vv\dr}\right].
\end{align}
Similarly, the same closure in the gradient form may be obtained from \eqref{eq:momentsgrd}:
\begin{equation}\label{eq:momentsgrdclosed}
	\pd_t\vv{\dr}(\vv z;t)\,=\,\frac{{4\vv\dr}^2}{\dr^2} \left[ 1- \-\frac{2\dr}{\Abes(\dr)}\right] \,\VD{\nNrg}{\vv\dr}\,-\,4\VD{\nNrg}{\bar{\vv\dr}}.
\end{equation}
This equation is quite similar to the canonical Landau-de Gennes equation  (or Ginzburg-Landau equation) for the free energy dissipation in $L_2$ metric,
\begin{equation}\label{eq:LdGGL}
	\pd_t\vv{\dr}(\vv z;t)\,=\,-\VD{\nNrg}{\bar{\vv\dr}}.
\end{equation}
Curiously, \eqref{eq:momentsgrdclosed} becomes exactly \eqref{eq:LdGGL} (up to a time-scale change), when $\dr(\vv z)=\dr_\eq$; $\gamma=2$ and $\Abes(\dr_\eq)=2\dr_\eq$, which corresponds to the isotropic-nematic phase transition. This is exactly when the Landau expansion of the free energy is valid. We would like to stress though, that in general, this maximal entropy closure is only mathematically justifiable when the relative entropy term in (\ref{eq:nrg_decomp}) is penalized in some appropriate asymptotic limit. In such a limit, however, the dynamics prescribed by equations \eqref{eq:nclosed} or \eqref{eq:momentsgrdclosed} itself becomes singular and reduces to vortex dynamics, as explained in this work.
%
\appendix 
%
\section{Notation and some useful facts}\label{sec:somp}
We use {\bf bold face} font to denote complex-valued functions and variables; regular font for their absolute values, e.g., $z=|\vv z|$. Given $\vv z=x+\mi y$, the operators $\pd_{\vv z}$ and $\pd_{\bar{\vv z}}$ are defined as
\begin{equation}
	\pd_{\vv z}\,=\,\frac{1}{2}\big(\pd_x\,-\,\mi\pd_y\big),\qquad\pd_{\bar{\vv z}}\,=\,\frac{1}{2}\big(\pd_x\,+\,\mi\pd_y\big).
\end{equation}
The complex form of Stokes' theorem may be written as
\begin{equation}\label{eq:cintbyparts}
	\oint_{\pd\dom} f(\vv z,\bar{\vv z})\md{\vv z}\,=\,2\mi\int_\dom\pd_{\bar{\vv z}}f(\vv z,\bar{\vv z})\dvz,
\end{equation}
where integral on the left is counter-clockwise contour integral, and integral on the right is the usual area integral, i.e., $\dvz=\md x\md y$.
\paragraph{Useful identities} involving $\Ln\vv z\,=\,\ln z+\mi\arg\vv z$:
\begin{subequations}
\begin{align}
		\pd_{\vv z}\Ln\vv z\,&=\,\frac{1}{\vv z}\,=\,2\pd_{\vv z}\ln z\,=\,2\mi\pd_{\vv z}\arg\vv z;\qquad\pd_{\bar{\vv z}}\Ln\vv z\,=\,0;\\
		\md\,\arg\vv z\,&=\,\pd_{\vv z}\arg\vv z\md\vv z\,+\,\pd_{\bar{\vv z}}\arg\vv z\md\bar{\vv z}\,=\,\frac{x}{z}\md y\,-\,\frac{y}{z}\md x.
\end{align}
\end{subequations}
Note that integration with $\md\,\arg\vv z$ is well-defined in $\C\setminus\{0\}$, even though $\arg\vv z$ is a multivalued function with jump discontinuities on closed contours encircling the origin.

\separate 

\noindent Let us state a lemma which is used in estimation of some integrals; its proof is straightforward.
\begin{lemma}\label{lemma1}
\ Let $\dom \subset \C$ and $B_{R_\eps} (\vz0)$ be a disk of radius $R_\eps=\bO(\eps)$ centered at $\vz0$. Assume that $\mu(\eps) \to 0$ as $\eps \to 0$ and a sequence $f^\eps(z)$ satisfies $(1- \mu(\eps))\leq |f^\eps(\vv z)| \leq 1$ for all $\vv z\in\dom\setminus B_{R_\eps} (\vz0)$.  Then 
\begin{equation}
	\int_{\dom \setminus B_{R_\eps} (\vz0)} f^\eps(\vv z) \frac{\md v(\vv z)}{| \vv z - \vz0|^2} = - 2 \pi \ln \eps + \bO(1).
\end{equation}
\end{lemma}

\separate

\paragraph{Special functions.}\label{sec:wgamma}

\begin{figure} 
\begin{center}
\raisebox{1mm}{\scalebox{0.45}{\includegraphics{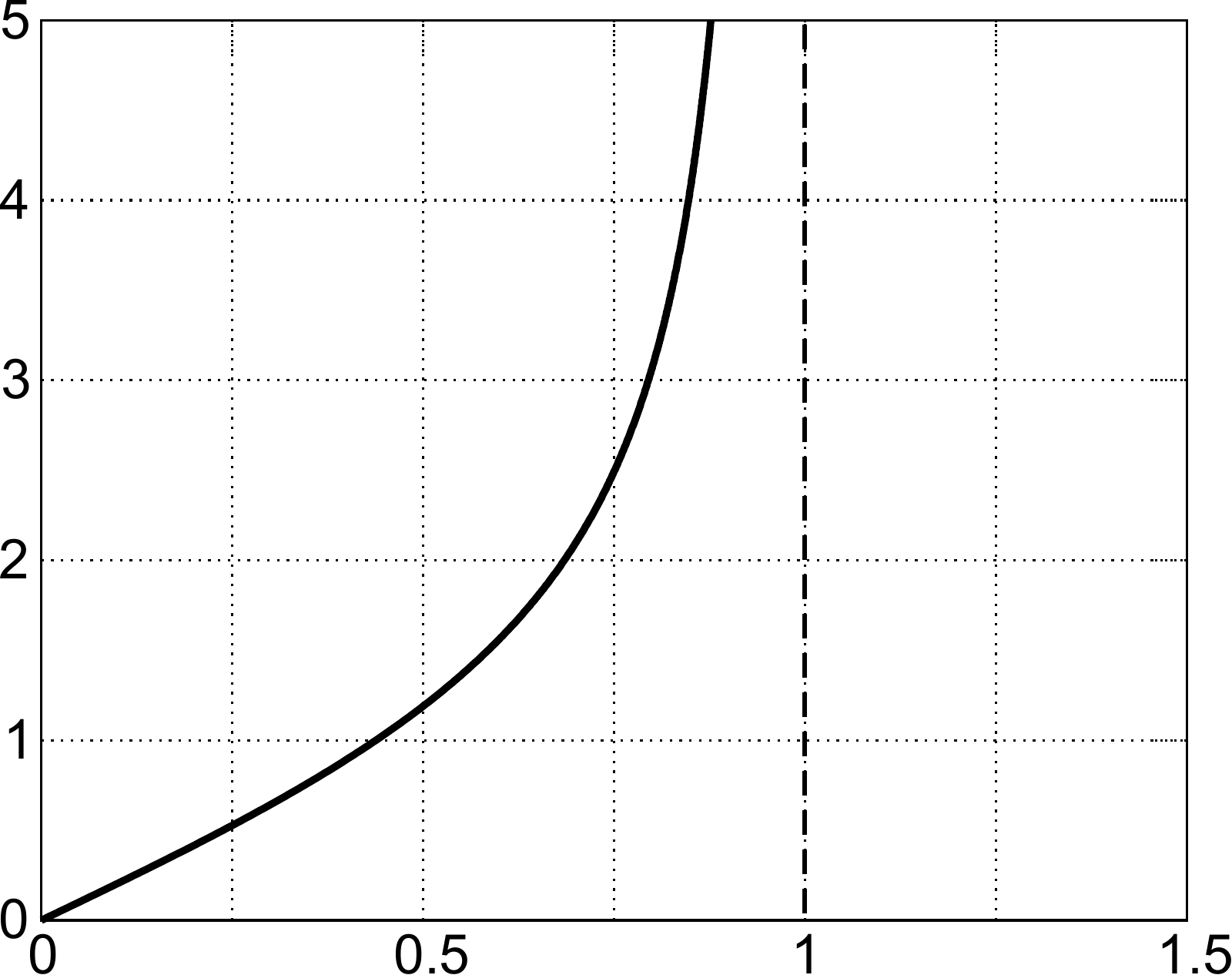}}}\hfill\scalebox{0.45}{\includegraphics{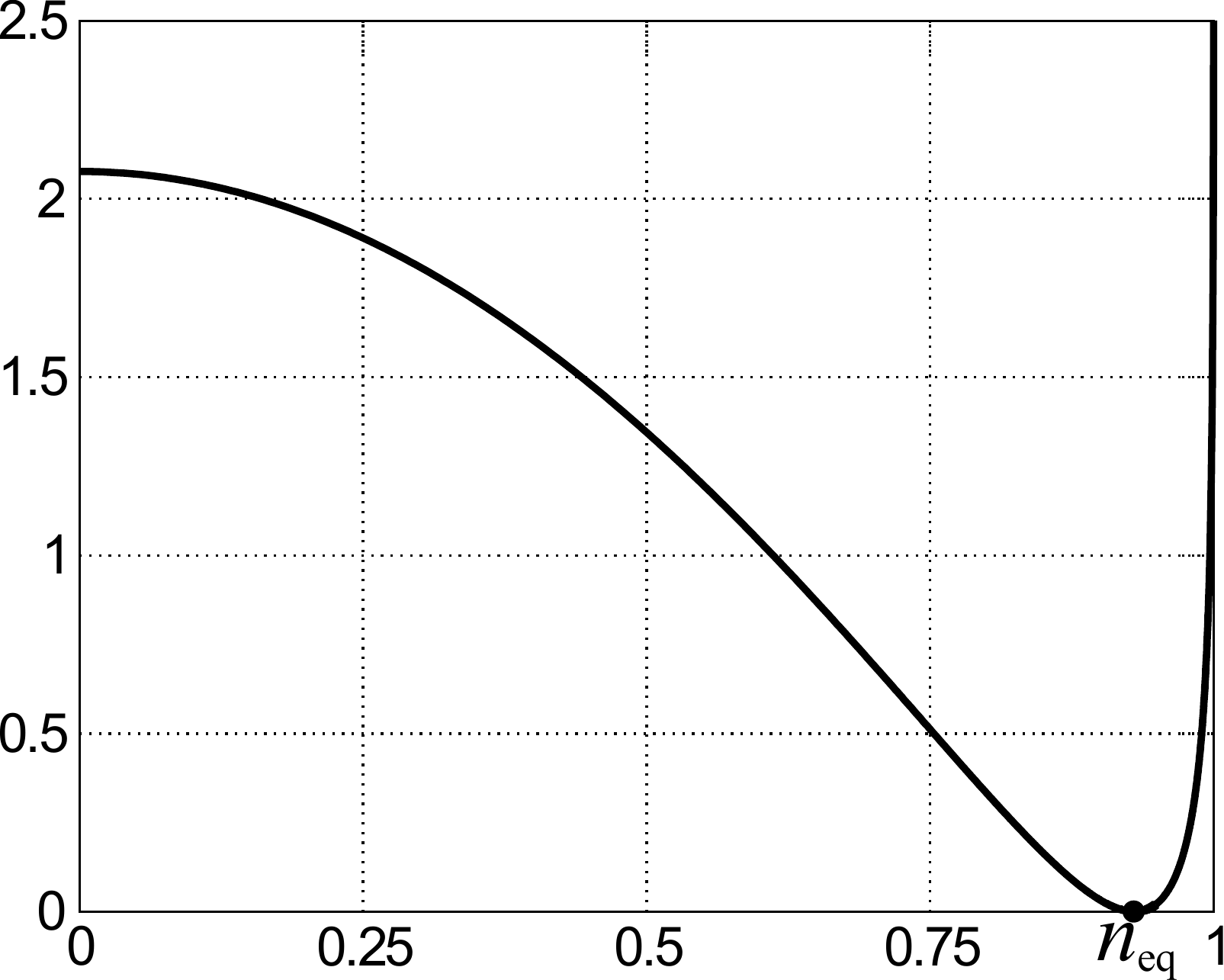}}
\end{center}
\caption{Graphs of the function $\Abes(n)$ (left) and the potential $W^\gamma(n)$ for $\gamma=6$ (right).}
\label{fig:graphs}
\end{figure}

In our work we use several special functions, such as the modified Bessel functions of the first kind, $\Ibes_{\nu}(\lambda)$, and the  function $\Abes(\dr)$. Here is a brief summary of their properties.
The function $\Abes (\dr)$ is the inverse of $\Ibes_1/\Ibes_0$, i.e., 
\begin{equation}
\frac{\Ibes_1\big(\Abes(\dr)\big)}{\Ibes_0\big(\Abes(\dr)\big)}\,=\,\dr.
\end{equation}
Using the properties of modified Bessel functions, it is straightforward to show that $\Abes (\dr)$ is a monotone increasing function defined on $(-1,1)$ with vertical asymptotes at $\dr =\pm1$; it is odd and convex when $\dr>0$. The graph of $\Abes (n)$ is shown in Figure~\ref{fig:graphs}. By direct computation we can verify that 
\begin{equation}\label{eq:app_ineqrad}
0 < 1\,-\,\frac{1}{\Ibes^2_0(\Abes(\dr))}\,\leq\,\min(1,C \dr^2),
\end{equation}
for some $C >0$ independent of $\dr$. This inequality is used in estimation of some of the integrals occurring in this paper.

The potential $W^\gamma(\dr)$ is given by
\begin{equation}\label{eq:pot_n1}
	W^{\gamma}(\dr)\,=\,-\,\frac{\gamma\dr^2}{2}\,+\,
	\big[\dr\Abes(\dr)\,-\,\ln\Ibes_0(\Abes(\dr))\big]\,+\,C_{\gamma},
\end{equation}
where $C_\gamma$ is chosen so that $W^\gamma (\dr) \geq 0$ with equality achieved at $\dr=\dr^\gamma_{\eq}$. The value of $\dr^\gamma_{\eq}$ satisfies $\smash{\gamma \dr^\gamma_\eq = \Abes \big(\dr^\gamma_\eq\big)}$. This equation has a nonzero solution for $\gamma>2$, which corresponds to the isotropic-nematic phase transition. The graph of $W^\gamma (\dr)$ is shown in Figure~\ref{fig:graphs}.

\separate

\noindent Here is a lemma which we use to invert the matrix $G$ in Section~\ref{sec:Doi}:
\begin{lemma}\label{lemma}
Let $\vv A$ be a symmetric block-matrix , representable, when $\delta\to0$, as
\begin{equation}
	\vv A\,=\,\frac{1}{\delta}\left[\begin{array}{cc}
		A^{}_{11}&0\\
		0~&0
	\end{array}
	\right]\,+\,\left[\begin{array}{cc}		
		B^{}_{11}&B^{}_{12}\\
		B^{}_{21}&B^{}_{22}
	\end{array}
	\right]\,+\,\delta\left[\begin{array}{cc}
		C^{}_{11}&C^{}_{12}\\
		C^{}_{21}&C^{}_{22}
	\end{array}
	\right]\,+\,\bO\big(\delta^2\big),
\end{equation}
where $A_{ii}$, $B_{ii}$, $C_{ii}$ symmetric matrices; $A_{11}$ and $B_{22}$ are invertible; $B^{}_{12}=B^\dag_{21}$, $C^{}_{12}=C^\dag_{21}$. Then its inverse, $\smash{\vv A^{-1}}$, exists for all sufficiently small $\delta$, and is given by
\begin{equation}\label{eq:app_mx1}
	\vv A^{-1}\,=\,
	\left[\begin{array}{cc}		
		0&0\\
		0&B^{-1}_{22}
	\end{array}
	\right]\,+\,\delta\left[\begin{array}{cc}
		A^{-1}_{11}&D^{}_{12}\\
		D^{}_{21}&D^{}_{22}
	\end{array}
	\right]\,+\,\bO\big(\delta^2\big),
\end{equation}
where $D^{}_{12}=-A^{-1}_{11}B^{}_{12}B^{-1}_{22}=D^\dag_{21}$, $D^{}_{22}=B^{-1}_{22}\big(C^{}_{22}-B^{}_{21}A^{-1}_{11}B^{}_{12}\big)B^{-1}_{22}$.

\end{lemma}

\begin{proof}
Computing the determinant of $\vv A$ via expansion with respect to rows corresponding to $A_{11}$, we obtain an asymptotic formula, $\smash{\det\vv A=\delta^{-d}\det A_{11}\det B_{22}+\bO\big(\delta^{1-d}\big)}$, where $d$ is the dimension of $A_{11}$. Because $A_{11}$ and $B_{22}$ are invertible, $\smash{\det\vv A^{-1}}\neq0$ for all sufficiently small $\delta$, and thus, within this range of $\delta$, $\vv A^{-1}$ exists.

Let us verify equation (\ref{eq:app_mx1}). Consider
\begin{equation}\label{eq:app_mx3}
	\vv B\,=\,
	\left[\begin{array}{cc}		
		0&0\\
		0&B^{-1}_{22}
	\end{array}
	\right]\,+\,\delta\left[\begin{array}{cc}
		A^{-1}_{11}&D^{}_{12}\\
		D^{}_{21}&0
	\end{array}
	\right].
\end{equation}
By direct computation we obtain,
\begin{equation}
	\vv A\vv B\,=\,\vv I\,+\,\delta\vv C
	\,+\,\bO(\delta^2);\qquad
	\vv C\,=\,
	\left[
	\begin{array}{cc}
		A^{-1}_{11}&B^{}_{11}D^{}_{12}+C^{}_{12}B^{-1}_{22}\\
		0&B^{}_{21}D^{}_{12}+C^{}_{22}B^{-1}_{22}
	\end{array}
	\right].
\end{equation}
Multiplying (on the left) both sides by $\vv A^{-1}$, we get
\begin{equation}\label{eq:app_mx2}
	\vv B\,=\,\vv A^{-1}\left\{\vv I\,+\,\delta\vv C\,+\,\bO(\delta^2)\right\}.
\end{equation}
By Gershgorin circle theorem, all eigenvalues of the matrix in curly brackets in equation (\ref{eq:app_mx2}) lie within $\bO(\delta)$ distance of 1, thus it is invertible, and its inverse is given up to $\bO\big(\delta^2\big)$ by $\vv I-\delta \vv C$. Therefore, $\smash{\vv A^{-1}\,=\vv B\big(\vv I-\delta \vv C+\bO(\delta^2)\big)\,}$, verifying the claim.
\end{proof}

%

\bibliographystyle{plain} 
\bibliography{../bibliography/bibl}

\begin{thebibliography}{10}

\bibitem{Alic12}
R.~Alicandro and M.~Ponsiglione.
\newblock {G}inzburg-{L}andau functionals and renormalized energy: A revised
  {$\Gamma$}-convergence approach.
\newblock {\em preprint}, 2011.

\bibitem{Ambro05}
L.~Ambrosio, N.~Gigli, and G.~Savar{\'e}.
\newblock Gradient flows in metric spaces and in the space of probability
  measures.
\newblock {\em Lectures in Mathematics, Birkh{\"a}user Verlag,
  Basel-Boston-Berlin}, 2005.

\bibitem{BerEdw}
A.N. Beris and B.J. Edwards.
\newblock Thermodynamics of flowing systems with internal microstructure.
\newblock {\em Oxford University Press}, 1994.

\bibitem{Beth94}
F.~Bethuel, H.~Brezis, and F.~Helein.
\newblock {G}inzburg-{L}andau vortices.
\newblock {\em Progress in nonlinear differential equations and their
  applications, Birkh{\"a}user, Boston}, 13, 1994.

\bibitem{dGen95}
P.~G. de~Gennes and J.~Prost.
\newblock The physics of liquid crystals.
\newblock {\em Clarendon Press, Oxford}, 1995.

\bibitem{DoiEdw}
M.~Doi and S.~F. Edwards.
\newblock The theory of polymer dynamics.
\newblock {\em Clarendon Press}, 1999.

\bibitem{E94}
W.~E.
\newblock Dynamics of vortices in {G}inzburg-{L}andau theories with
  applications to superconductivity.
\newblock {\em Physica D}, 77(4):383--404, 1994.

\bibitem{E06}
W.~E and P.~Zhang.
\newblock A molecular kinetic theory of inhomogeneous liquid crystal flow and
  the small deborah number limit.
\newblock {\em Methods and Applications of Analysis}, 13(2):181--198, 2006.

\bibitem{Erick61}
J.~L. Ericksen.
\newblock Conservation laws for liquid crystals.
\newblock {\em Journal of Rheology}, 5(1):23--34, 1961.

\bibitem{Fat07a}
I.~Fatkullin and V.~Slastikov.
\newblock On spatial variations of nematic ordering.
\newblock {\em Physica D}, 237(20):2577--2586, 2008.

\bibitem{Fat09a}
I.~Fatkullin and V.~Slastikov.
\newblock Vortices in two-dimensional nematics.
\newblock {\em Communications in Mathematical Sciences}, 7(4):917--938, 2009.

\bibitem{Lesl68}
F.M. Leslie.
\newblock Some constitutive equations for liquid crystals.
\newblock {\em Archive for Rational Mechanics and Analysis}, 28:265, 1968.

\bibitem{Lin96}
F.~H. Lin.
\newblock Some dynamical properties of {G}inzburg-{L}andau vortices.
\newblock {\em Communications on Pure and Applied Mathematics}, 49(4):323--359,
  1996.

\bibitem{Neu90}
J.~C. Neu.
\newblock Vortices in complex scalar fields.
\newblock {\em Physica D}, 43:385--406, 1990.

\bibitem{San98}
E.~Sandier.
\newblock Lower bounds for the energy of unit vector fields and applications.
\newblock {\em Journal of Functional Analysis}, 152(2):379--403, 1998.

\bibitem{San04}
E.~Sandier and S.~Serfaty.
\newblock Gamma-convergence of gradient flows with applications to
  {G}inzburg-{L}andau.
\newblock {\em Communications on Pure and Applied Mathematics},
  57(12):1627--1672, 2004.

\bibitem{Vill03}
C.~Villani.
\newblock Topics in optimal transportation.
\newblock {\em Graduate Studies in Mathematics, AMS}, 58, 1992.

\bibitem{Wang13}
P.~Zhang W.~Wang and Z.~Zhang.
\newblock The small deborah number limit of the doi-onsager equation to the
  ericksen-leslie equation.
\newblock {\em arXiv:1206.5480}, 2013.

\end{thebibliography}
\end{document}

%